\documentclass[A4,12pt]{article}
\usepackage{amsmath,amssymb,amsfonts,amsthm,graphicx}
\usepackage{upgreek}
\usepackage{multirow}
\usepackage{diagbox}
\usepackage{color}
\usepackage{cite}
\usepackage{epsfig}
\usepackage{epstopdf}
\usepackage{footnote}
\usepackage{cite}
\usepackage{caption}
\usepackage{subcaption}
\usepackage{todonotes}
\usepackage{graphicx}
\usepackage{enumerate}
\usepackage{times}
\usepackage{algorithm}
\usepackage{algpseudocode}
\usepackage{tikz}
\usepackage{tkz-berge}
\usepackage{tkz-graph}
\usetikzlibrary{graphs}
\usetikzlibrary{graphs.standard}

\usepackage{tkz-berge}
\usepackage[top=2cm, bottom=2cm, left=2cm, right=2cm]{geometry}

\newcommand{\gmd}{\gamma_{md}}

\newtheorem{theorem}{Theorem}[section]
\newtheorem{lemma}[theorem]{Lemma}

\theoremstyle{definition}
\newtheorem{definition}[theorem]{Definition}

\newtheorem{corollary}[theorem]{Corollary}

\theoremstyle{remark}
\newtheorem{remark}[theorem]{\bf Remark}
\tikzstyle{vertex}=[circle, draw, inner sep=0pt, minimum size=6pt]
\newcommand{\vertex}{\node[vertex]}

\title{On Mixed Domination in Generalized Petersen Graphs}

\author{M. Rajaati$^{1}$, M. R. Hooshmandasl$^{2}$,  M. Alambardar Meybodi$^{3}$ and B. Davvaz$^{4}$  \\
\footnotesize{$^{1,2,3}$Department of Computer Science, Yazd University, Yazd, Iran.}\\
\footnotesize{$^{4}$Department of Mathematics, Yazd University, Yazd, Iran.}\\
\footnotesize{$^{1,2,3}$The Laboratory of Quantum Information Processing, Yazd University, Yazd, Iran.}  \\
\footnotesize{e-mail: $^1$m.rajaati@stu.yazd.ac.ir, $^2$hooshmandasl@yazd.ac.ir, $^3$m.alambardar@stu.yazd.ac.ir, $^4$davvaz@yazd.ac.ir} }
\date{}

\begin{document}
\maketitle

\begin{abstract}		
	
	Given a graph $G = (V, E)$, a set $S \subseteq V \cup E$ of vertices and edges is called a mixed dominating set if every vertex and edge that is not included in $S$ happens to be adjacent or incident to a member of $S$. The mixed domination number $\gmd(G)$ of the graph  is the size of the smallest mixed dominating set of $G$.
	We present an explicit method for constructing optimal mixed dominating sets in Petersen graphs $P(n, k)$ for $k \in \{1, 2\}$. Our method also provides a new upper bound for other Petersen graphs.
	
\end{abstract}
\section{Introduction and Preliminaries}

In a graph $G = (V, E)$, a set $S \subseteq V$ is called a dominating set if every vertex is either in $S$ or adjacent to a vertex in $S$. The size of the smallest dominating set of a graph $G$ is denoted by $\gamma(G)$ and every dominating set of this size is called a $\gamma$-set or optimal dominating set of $G$.

Domination is a natural model for many problems in operations research such as optimal placement of facilities like fire stations and hospitals. Minor differences in the characteristics of its many applications has led to various different extensions of the notion of domination. See \cite{hedetniemi1995domination} for a survey.
A well-studied example of these generalizations is mixed domination, which is motivated by a formulation of the problem of ensuring reliability in electric power networks \cite{zhao2011algorithmic}. 

Formally, a set $S \subseteq V \cup E$ of vertices and edges of a graph $G=(V, E)$ is called a mixed dominating set if every element $x \in (V \cup E) \setminus S $ is either adjacent or incident to an element of $S$.  The mixed domination number of $G$  is the size of the smallest mixed dominating set of $G$ and is denoted by $\gmd(G)$. 

The problem of computing $\gmd(G)$ is called the mixed domination problem, or more traditionally total cover problem, and was first introduced in 1977 \cite{alavi1977total}.
Its corresponding decision problem is NP-complete \cite{majumdar1992neighborhood}, and remains so even when restricted to bipartite and chordal graphs \cite{hedetniemi1995domination} or planar bipartite graphs of maximum degree four \cite{manlove1999algorithmic}. On the other hand, linear algorithms are known for some classes of graphs like trees and cacti \cite{zhao2011algorithmic,lan2013mixed}.

The generalized Petersen graph $P(n, k)$ with $k \leq \frac{n}{2}$ is defined to be a graph on $2n$ vertices with $V(P(n, k)) := \{v_i, u_i ~\vert~ 0 \leq i \leq n-1\}$ and $E(P(n, k)) := \{v_iv_{(i+1) \mod n}, v_iu_i, u_iu_{(i+k) \mod n} ~\vert~ 0 \leq i \leq n-1\}$. Intuitively, the $P(n, k)$ graph consists of an internal and an external part, each with $n$ vertices. The external part forms a cycle. In the internal part, there is an edge connecting each vertex to $k$ vertices ahead. Moreover, each vertex in the internal part is connected to its corresponding vertex in the external part. For example $P(10, 3)$ is illustrated in Figure~\ref{fig:sample}.

\begin{figure}
	\centering
	\resizebox*{.25\textwidth}{!}{
		\begin{tikzpicture}[rotate=90]
		\tikzstyle{VertexStyle} = [shape = classic]
		\tikzstyle{EdgeStyle}= [thin,double= black,double distance= 2pt]
		%\SetVertexNoLabel
		\SetVertexNormal[LineWidth=1pt]
		\grGeneralizedPetersen[Math,RA=5]{10}{3}
		
		\begin{scope}[]
		\node[shape=circle,draw=white] (1) at (5.75,0) {\LARGE $ v_0$};
		\node[shape=circle,draw=white] (2) at (4.5,-3.5) {\LARGE $v_1$};
		\node[shape=circle,draw=white] (3) at (1.75,-5.5) {\LARGE $v_2$};
		\node[shape=circle,draw=white] (4) at (-1.75,-5.5) {\LARGE $v_3$};
		\node[shape=circle,draw=white] (5) at (-4.5,-3.5) {\LARGE $v_4$};
		\node[shape=circle,draw=white] (6) at (-5.75,0) {\LARGE $v_5$};
		\node[shape=circle,draw=white] (7) at (-4.5,3.5) {\LARGE $v_6$};
		\node[shape=circle,draw=white] (8) at (-1.75,5.5) {\LARGE $v_7$};
		\node[shape=circle,draw=white] (9) at (1.75,5.5) {\LARGE $v_8$};
		\node[shape=circle,draw=white] (10) at (4.5,3.5) {\LARGE $v_9$};
		\node[shape=circle,draw=white] (11) at (2.25,0) {\LARGE $u_0$};.
		\node[shape=circle,draw=white] (12) at (1.8,-1.35) {\LARGE $u_1$};
		\node[shape=circle,draw=white] (13) at (0.75,-2.15) {\LARGE $u_2$};
		\node[shape=circle,draw=white] (14) at (-0.75,-2.15) {\LARGE $u_3$};
		\node[shape=circle,draw=white] (15) at (-1.8,-1.35) {\LARGE $u_4$};
		\node[shape=circle,draw=white] (16) at (-2.25,0) {\LARGE $u_5$};
		\node[shape=circle,draw=white] (17) at (-1.8,1.35) {\LARGE $u_6$};
		\node[shape=circle,draw=white] (18) at (-0.75,2.15) {\LARGE $u_7$};
		\node[shape=circle,draw=white] (19) at (0.75,2.15) {\LARGE $u_8$};
		\node[shape=circle,draw=white] (20) at (1.8,1.35) {\LARGE $u_9$};
		\end{scope}
		\end{tikzpicture}}
	\caption{P(10,3)}  
	\label{fig:sample}
\end{figure}

Domination and its variations have been studied extensively for generalized Petersen graphs $P(n,k)$ in recent years\cite{behzad2008domination,ebrahimi2009vertex,fu2009domination,yan2009exact}. It was shown in \cite{fu2009domination} that $\gamma(P(n,2))=n-\lfloor 
\frac{n}{5} \rfloor+\lceil \frac{n+2}{5}\rceil$. Yan et al. \cite{yan2009exact} proved $\gamma(P(2k+1,k))= \left \lceil \frac{3(2K+1)}{5}\right  \rceil $. However, similar results had not been established for mixed domination.
In the next sections, we propose a constructive method to obtain small mixed dominating sets in generalized Petersen graphs $P(n, k)$. Our approach is optimal for $k \leq 2$ and provides a new upper-bound for other cases.

%The present article is organized as follows: In Section \ref{sec:prem}, we define our notation and present basic lemmas. In Section \ref{sec:design}, we design some patterns for obtaining small mixed dominating sets in $P(n,k)$, where $k\in \{1,2\}$ and prove that they produce optimal sets. Section \ref{sec:up} discusses a general upper bound for the mixed domination number of $P(n,k)$ where $k\geq 3$. Finally, Section \ref{sec:conclusion} concludes the paper with a series of open problems and suggestions for future research.

We now define our notation and present a basic lemma that will be used in the rest of the paper.
Let $\upxi \in  V \cup E$ be an element in $G$. Its closed mixed neighborhood $N_m[\upxi] \subseteq V \cup E$
consists of $\upxi$ itself and all other vertices and edges that are adjacent or incident to $\upxi$. We say that $\upxi'$ is mixed dominated by $\upxi$ iff $\upxi' \in N_m[\upxi]$.
We also define the neighborhood of a set of elements as the union of the neighborhoods of its members. Hence, a set $S \subseteq V \cup E$ is a mixed dominating set iff $N_m[S] = V \cup E$. It is possible that some elements are dominated more than once by members of $S$. We introduce the notion of redomination number to capture such overlaps.  

\begin{definition}
	Let $S$ be a mixed dominating set of $G =(V, E)$. The redominaton number of an element $\upxi \in V\cup E$ is defined as
	$$rd_S(\upxi):=\vert N_m[\upxi]\cap S\vert-1$$
	also, for $X\subseteq V\cup E$ we set
	\[rd_S(X):=\sum_{\upxi \in X}rd_S(\upxi)\]
\end{definition}

For $n \geq 4$, each element of the graph $P(n,k)$ dominates exactly $7$ elements (its $6$ neighbors and itself). Given that $P(n, k)$ has $5n$ elements, i.e.~$2n$ vertices and $3n$ edges, it is clear that $\gmd(P(n,k))\geq \frac{5n}{7}$. Equality holds only if there exists a mixed dominating set $S$ such that every element of $P(n,k)$ is dominated exactly once. In this case, we would have $7\vert S\vert-5n=0$. This is clearly impossible, hence the inequality is always strict and our goal is to find a mixed dominating set $S$ of $P(n,k)$ that minimizes the difference $rd_S(V\cup E)=7|S|-5n$. This discussion naturally leads to the following lemma:

\begin{lemma}
	Let $S$ be an optimal mixed dominating set of $G$, then $\gmd(P(n,k)) = \vert S \vert =   \frac{5n+rd_S(V \cup E)}{7} .$
\end{lemma}

Let $t$ be a natural number. We partition the vertex set of $P(n,k)$ into blocks of size $2t$ or less, each containing at most $t$ consecutive internal vertices and their corresponding external ones. We then call $t$ the partitioning factor. Formally, the $i$-th block for $0 \leq i \leq \lfloor \frac{n}{t} \rfloor$ consists of the following vertices:
$$V_i^t := {\{v_{ti+j}, u_{ti+j} ~\vert~ 0 \leq j < t,~~ ti+j < n\}}.$$

Corresponding to $V_i^t$, we define  $E^t_i$, $E^t_{i,i+1}$ and $G^t_i$ as follows: 
\begin{align*}
E^t_i &= \{ uv \in E ~\vert~ u,v \in V^t_i  \},\\
E^t_{i,i+1} &= \{ uv \in E ~\vert~ u \in V^t_i, v \in V^t_{i+1} \},\\
G[V_i^t] &=(V_i^t, E_i^t).
\end{align*}

Basically, $E_i^t$ is the set of edges between vertices in $V_i^t$, $E^t_{i, i+1}$ contains edges between $V_i^t$ and $V_{i+1}^t$ and $G[V_i^t]$ is the subgraph of $G$ induced by $V_i^t$.
\section{Mixed Domination in $P(n,k)$}
In this section we obtain exact formulas for the mixed domination number of generalized Petersen graphs $P(n, k)$ when $k \leq 2$ and establish an upper bound for the case where $k\geq 3$.
\subsection{First Case: k = 1}

Using a brute-force approach, we obtained optimal mixed dominating sets for $P(n, 1)$ where $n \leq 7$. % as shown in Figure~\ref{fig-P(1n7,1)}.
Table~\ref{Mixed1-7} summarizes the mixed domination numbers of these graphs.

\begin{center}
	\begin{table}[h]
		\caption{number of needed elements to dominate}  \label{Mixed1-7}	
		\centering
		\begin{tabular}{|c|c|c|c|c|c|c|c|}
			\hline
			$n$ & 1 & 2 & 3 & 4 & 5 & 6& 7\\
			\hline
			$\gamma_{md}(P(n,1))$ & 1 & 2 & 3 & 4 & 4 & 5 & 6 \\
			\hline
		\end{tabular}  
	\end{table}
\end{center}

We now turn to the case where $n \geq 8$. In the following theorem, we set the partitioning factor $t$ to $8$ and let $m := \lfloor \frac{n}{8} \rfloor $ and $\bar{r} := n \mod 8$.

\begin{theorem}\label{Domination(n,1)}
	The mixed domination number of $P(n,1)$ where $n\geq 8$ is
	\begin{equation}\label{Eq-1}
	\gmd(P(n,1))=\left\{\begin{array}{llll}
	6m & & \text{if } & \bar{r}=0,  \\ \noalign{\medskip}
	6m+2 & & \text{if }& \bar{r}=1,2,  \\ \noalign{\medskip}
	6m+3 & &\text{if }& \bar{r}=3,  \\ \noalign{\medskip}
	6m+4 & &\text{if }& \bar{r}=4,5,  \\ \noalign{\medskip}
	6m+5 & &\text{if }& \bar{r}=6,  \\ \noalign{\medskip}
	6m+6 & &\text{if }& \bar{r}=7,  \\ \noalign{\medskip}
	\end{array}\right.
	\end{equation}
\end{theorem}

\begin{proof}
	
	We first prove that the given values are an upper-bound for $\gmd(P(n, 1))$ and then prove that they are also a lower-bound.
	To settle the former, we construct a mixed dominating set of the desired size.
	
	We first handle the vertices and edges in each of the full blocks of the form $V^8_i$, i.e.~those blocks that contain exactly $16$ vertices. We use the following set, noted as $M$, as the subset of our mixed dominating set that lies within these full blocks:
	\[
	M=\{u_{8i},v_{8i+1}v_{8i+2},u_{8i+2}u_{8i+3},v_{8i+4},u_{8i+5}u_{8i+6}, v_{8i+6}v_{8i+7} :  0\leq i\leq m-1\}\]
	
	Figure~\ref{fig-P(10,1)-block8} provides an example of this set in $P(10, 1)$.
	
	\newcommand{\cercle}[6]{
		\node[circle,inner sep=0,minimum size={2*#2}](a) at (#1) {};
		\draw[#6,line width=#5] (a.#3) arc (#3:{#3+#4}:#2);
	}
	\begin{center}
		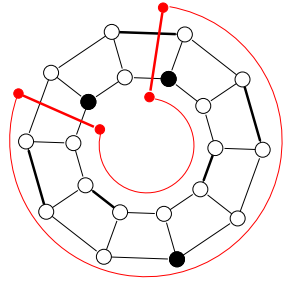
\begin{figure}[h]
			\centering
			\resizebox*{.3\textwidth}{!}{
				\begin{tikzpicture}[rotate=70,-,>=stealth',shorten >=2pt,auto,node distance=4cm,main node/.style={circle,draw,font=\sffamily\Large\bfseries}]
				
				\tikzstyle{VertexStyle} = [shape = classic]
				\tikzstyle{EdgeStyle}= [thin,double= black,double distance= 2pt]
				\SetVertexNoLabel
				\SetVertexNormal[LineWidth=1pt]
				\grGeneralizedPetersen[Math,RA=5]{10}{1}
				\AddVertexColor{black}{b0,a6,b2}
				\draw[draw=black,line width=3pt]
				(a9) to (a8);
				\draw[draw=black,line width=3pt]
				(b8) to (b7);
				\draw[draw=black,line width=3pt]
				(b5) to (b4);
				\draw[draw=black,line width=3pt]
				(a4) to (a3);
				\draw[draw=black,line width=3pt]
				(a1) to (a0);
				
				\node[main node,draw=red,fill=red] at (5.75,1.25) (1) {};
				\node[main node,draw=red,fill=red] at (2,0.5) (2)  {};
				\node[main node,draw=red,fill=red] at (0.25,5.75) (3) {};
				\node[main node,draw=red,fill=red] at (0,2) (4)  {};
				\draw[draw=red,line width=3pt]
				(1) to (2);
				\draw[draw=red,line width=3pt](3) to (4);
				%\draw[<-] (3) to [bend left] node [above right] {} (1);
				%\draw[<-] (4) to [bend left] node [above right] {} (2);
				%\draw[<-] (5) to [bend left] node [above right] {} (3);
				%\draw[<-] (6) to [bend left] node [above right] {} (4);
				
				\coordinate (OR) at (3, 6);
				\cercle {OR} {5.75cm} {10} {-280} {1.00} {red};
				
				\coordinate (OR) at (1, 2);
				\cercle {OR} {2cm} {10} {-280} {1.00} {red};
				%\path[every node/.style={font=\sffamily\small}]
				%(1) edge [bend left] node[left] {} (4)
				\end{tikzpicture}}
			\caption{Elements of the set $M$ are shown in bold in the red area.}\label{fig-P(10,1)-block8}
		\end{figure}       
	\end{center}
	Now, we complement the set $M$ with additional elements from the final block to obtain the desired mixed dominating set $S$ as follows:
	\[
	S=\left\{\begin{array}{lll}
	M & \text{if } & \bar{r}=0,  \\ \noalign{\medskip}
	M\cup \{u_{8m},v_{8m}v_{0}\} & \text{if }& \bar{r}=1,  \\ \noalign{\medskip}
	M\cup \{u_{8m},v_{8m+1}v_{0}\} & \text{if }& \bar{r}=2,  \\ \noalign{\medskip}
	M\cup \{u_{8m},v_{8m+1}v_{8m+2},u_{8m+1}u_{8m+2} \} & \text{if }& \bar{r}=3,  \\ \noalign{\medskip}
	M\cup \{u_{8m}, v_{8m+1}v_{8m+2},u_{8m+2}u_{8m+3},v_{8m+3} \} & \text{if }& \bar{r}=4,  \\ \noalign{\medskip}
	M\cup \{u_{8m},v_{8m+1}v_{8m+2},u_{8m+2}u_{8m+3},v_{8m+4} \} & \text{if }& \bar{r}=5,  \\ \noalign{\medskip}
	M\cup \{u_{8m},v_{8m+1}v_{8m+2},u_{8m+2}u_{8m+3},v_{8m+4},u_{8m+5}v_{8m+5} \} & \text{if }& \bar{r}=6,  \\ \noalign{\medskip}
	M \cup \{u_{8m},v_{8m+1}v_{8m+2},u_{8m+2}u_{8m+3},v_{8m+4},u_{8m+5}u_{8m+6},v_{8m+6}v_{0} \}  & \text{if }& \bar{r}=7,  \\ \noalign{\medskip}
	\end{array}\right.
	\]
	It is easy to see that $S$ is a mixed dominating set with  the required cardinality. Figure~\ref{fig-P(n,1)} provides examples of our construction on $P(n, 1)$ for $8 \leq n \leq 15$.
	
	\begin{center}
		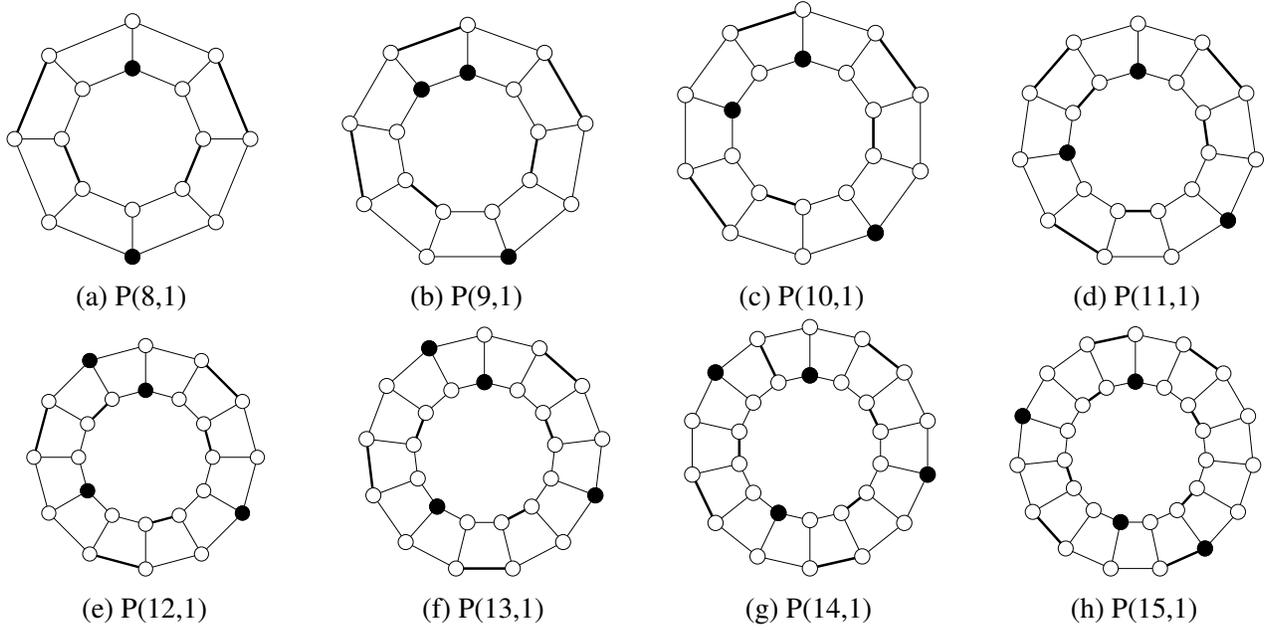
\begin{figure}[h!]
			\begin{subfigure}[b]{0.24\textwidth}
				\centering
				\resizebox*{.8\textwidth}{!}{
					\begin{tikzpicture}[rotate=90]
					\tikzstyle{VertexStyle} = [shape = classic]
					\tikzstyle{EdgeStyle}= [thin,double= black,double distance= 2pt]
					\SetVertexNoLabel
					\SetVertexNormal[LineWidth=1pt]
					\grGeneralizedPetersen[Math,RA=5]{8}{1}
					\AddVertexColor{black}{b0,a4}
					\draw[draw=black,line width=3pt]
					(a1) to (a2);
					\draw[draw=black,line width=3pt]
					(b2) to (b3);
					\draw[draw=black,line width=3pt]
					(b5) to (b6);
					\draw[draw=black,line width=3pt]
					(a6) to (a7);
					%					\begin{scope}[]
					%					\node[shape=circle,draw=white] (1) at (5.75,0) {$v_0$};
					%					\node[shape=circle,draw=white] (2) at (4,-4) {$v_1$};
					%					\node[shape=circle,draw=white] (3) at (0,-6) {$v_2$};
					%					\node[shape=circle,draw=white] (4) at (-4,-4) {$v_3$};
					%					\node[shape=circle,draw=white] (5) at (-5.75,0) {$v_4$};
					%					\node[shape=circle,draw=white] (6) at (-4,4) {$v_5$};
					%					\node[shape=circle,draw=white] (7) at (0,6) {$v_6$};
					%					\node[shape=circle,draw=white] (8) at (4,4) {$v_7$};
					%					\node[shape=circle,draw=white] (16) at (2,0) {$u_0$};
					%					\node[shape=circle,draw=white] (17) at (1.5,-1.5) {$u_1$};
					%					\node[shape=circle,draw=white] (18) at (0,-2) {$u_2$};
					%					\node[shape=circle,draw=white] (19) at (-1.5,-1.5) {$u_3$};
					%					\node[shape=circle,draw=white] (20) at (-2,0) {$u_4$};
					%					\node[shape=circle,draw=white] (21) at (-1.5,1.5) {$u_5$};
					%					\node[shape=circle,draw=white] (22) at (0,2) {$u_6$};
					%					\node[shape=circle,draw=white] (23) at (1.5,1.5) {$u_7$};
					%					\end{scope}
					\end{tikzpicture}}
				\caption{P(8,1)}             
			\end{subfigure}
			\begin{subfigure}[b]{0.24\textwidth}
				\centering
				\resizebox*{.8\textwidth}{!}{
					\begin{tikzpicture}[rotate=90]
					\tikzstyle{VertexStyle} = [shape = classic]
					\tikzstyle{EdgeStyle}= [thin,double= black,double distance= 2pt]
					\SetVertexNoLabel
					\SetVertexNormal[LineWidth=1pt]
					\grGeneralizedPetersen[Math,RA=5]{9}{1}
					\AddVertexColor{black}{b0,a5,b1}
					\draw[draw=black,line width=3pt]
					(a7) to (a8);
					\draw[draw=black,line width=3pt]
					(b6) to (b7);
					\draw[draw=black,line width=3pt]
					(b4) to (b3);
					\draw[draw=black,line width=3pt]
					(a3) to (a2);
					\draw[draw=black,line width=3pt]
					(a1) to (a0);
					\end{tikzpicture}}
				\caption{P(9,1)}          
			\end{subfigure}
			\begin{subfigure}[b]{0.24\textwidth}
				\centering
				\resizebox*{.8\textwidth}{!}{
					\begin{tikzpicture}[rotate=90]
					\tikzstyle{VertexStyle} = [shape = classic]
					\tikzstyle{EdgeStyle}= [thin,double= black,double distance= 2pt]
					\SetVertexNoLabel
					\SetVertexNormal[LineWidth=1pt]
					\grGeneralizedPetersen[Math,RA=5]{10}{1}
					\AddVertexColor{black}{b0,a6,b2}
					\draw[draw=black,line width=3pt]
					(a9) to (a8);
					\draw[draw=black,line width=3pt]
					(b8) to (b7);
					\draw[draw=black,line width=3pt]
					(b5) to (b4);
					\draw[draw=black,line width=3pt]
					(a4) to (a3);
					\draw[draw=black,line width=3pt]
					(a0) to (a1);
					\end{tikzpicture}}
				\caption{P(10,1)} \label{c}               
			\end{subfigure}
			\begin{subfigure}[b]{0.24\textwidth}
				\centering
				\resizebox*{.8\textwidth}{!}{
					\begin{tikzpicture}[rotate=90]
					\tikzstyle{VertexStyle} = [shape = classic]
					\tikzstyle{EdgeStyle}= [thin,double= black,double distance= 2pt]
					\SetVertexNoLabel
					\SetVertexNormal[LineWidth=1pt]
					\grGeneralizedPetersen[Math,RA=5]{11}{1}
					\AddVertexColor{black}{b0,a7,b3}
					\draw[draw=black,line width=3pt]
					(a10) to (a9);
					\draw[draw=black,line width=3pt]
					(b9) to (b8);
					\draw[draw=black,line width=3pt]
					(b6) to (b5);
					\draw[draw=black,line width=3pt]
					(a5) to (a4);
					\draw[draw=black,line width=3pt]
					(b2) to (b1);
					\draw[draw=black,line width=3pt]
					(a2) to (a1);
					\end{tikzpicture}}                                       
				\caption{P(11,1)}                
			\end{subfigure}
			\begin{subfigure}[b]{0.26\textwidth}
				\centering
				\resizebox*{.7\textwidth}{!}{
					\begin{tikzpicture}[rotate=90]
					\tikzstyle{VertexStyle} = [shape = classic]
					\tikzstyle{EdgeStyle}= [thin,double= black,double distance= 2pt]
					\SetVertexNoLabel
					\SetVertexNormal[LineWidth=1pt]
					\grGeneralizedPetersen[Math,RA=5]{12}{1}
					\AddVertexColor{black}{b0,a8,b4,a1}
					\draw[draw=black,line width=3pt]
					(a11) to (a10);
					\draw[draw=black,line width=3pt]
					(b10) to (b9);
					\draw[draw=black,line width=3pt]
					(b7) to (b6);
					\draw[draw=black,line width=3pt]
					(a6) to (a5);
					\draw[draw=black,line width=3pt]
					(a3) to (a2);
					\draw[draw=black,line width=3pt]
					(b2) to (b1);

					\end{tikzpicture}}
				\caption{P(12,1)}          
			\end{subfigure}
			\begin{subfigure}[b]{0.24\textwidth}
				\centering
				\resizebox*{.8\textwidth}{!}{
					\begin{tikzpicture}[rotate=90]
					\tikzstyle{VertexStyle} = [shape = classic]
					\tikzstyle{EdgeStyle}= [thin,double= black,double distance= 2pt]
					\SetVertexNoLabel
					\SetVertexNormal[LineWidth=1pt]
					\grGeneralizedPetersen[Math,RA=5]{13}{1}
					\AddVertexColor{black}{b0,a9,b5,a1}
					\draw[draw=black,line width=3pt]
					(a12) to (a11);
					\draw[draw=black,line width=3pt]
					(b11) to (b10);
					\draw[draw=black,line width=3pt]
					(b8) to (b7);
					\draw[draw=black,line width=3pt]
					(a7) to (a6);
					\draw[draw=black,line width=3pt]
					(a4) to (a3);
					\draw[draw=black,line width=3pt]
					(b3) to (b2);

					\end{tikzpicture}}
				\caption{P(13,1)}          
			\end{subfigure}
			\begin{subfigure}[b]{0.24\textwidth}
				\centering
				\resizebox*{.8\textwidth}{!}{
					\begin{tikzpicture}[rotate=90]
					\tikzstyle{VertexStyle} = [shape = classic]
					\tikzstyle{EdgeStyle}= [thin,double= black,double distance= 2pt]
					\SetVertexNoLabel
					\SetVertexNormal[LineWidth=1pt]
					\grGeneralizedPetersen[Math,RA=5]{14}{1}
					\AddVertexColor{black}{b0,a10,b6,a2}
					\draw[draw=black,line width=3pt]
					(a13) to (a12);
					\draw[draw=black,line width=3pt]
					(b12) to (b11);
					\draw[draw=black,line width=3pt]
					(b9) to (b8);
					\draw[draw=black,line width=3pt]
					(a8) to (a7);
					\draw[draw=black,line width=3pt]
					(a5) to (a4);
					\draw[draw=black,line width=3pt]
					(b4) to (b3);
					\draw[draw=black,line width=3pt]
					(a1) to (b1);
					
					\end{tikzpicture}}
				\caption{P(14,1)}          
			\end{subfigure}
			\begin{subfigure}[b]{0.24\textwidth}
				\centering
				\resizebox*{.8\textwidth}{!}{
					\begin{tikzpicture}[rotate=90]
					\tikzstyle{VertexStyle} = [shape = classic]
					\tikzstyle{EdgeStyle}= [thin,double= black,double distance= 2pt]
					\SetVertexNoLabel
					\SetVertexNormal[LineWidth=1pt]
					\grGeneralizedPetersen[Math,RA=5]{15}{1}
					\AddVertexColor{black}{b0,a9,b7,a3}
					\draw[draw=black,line width=3pt]
					(a14) to (a13);
					\draw[draw=black,line width=3pt]
					(b13) to (b12);
					\draw[draw=black,line width=3pt]
					(b10) to (b9);
					\draw[draw=black,line width=3pt]
					(a9) to (a8);
					\draw[draw=black,line width=3pt]
					(a6) to (a5);
					\draw[draw=black,line width=3pt]
					(b5) to (b4);
					\draw[draw=black,line width=3pt]
					(b2) to (b1);
					\draw[draw=black,line width=3pt]
					(a1) to (a0);
					
					\end{tikzpicture}}
				\caption{P(15,1)}          
			\end{subfigure}
			\caption{Optimal mixed dominating sets in $P(n,1)$, where $8\leq n \leq 15$. Elements of the mixed dominating set are shown in bold.}\label{fig-P(n,1)}
		\end{figure}       
	\end{center}

	To settle the other part, we need a more fine-grained partitioning. To simplify the proof process, we consider a partitioning factor of $4$. We prove that for any optimal mixed dominating set $S$, in every set of elements of the form $G[V_i^4]\cup E^4_{i,i+1}$ where $0\leq i \leq \lfloor \frac{n}{4}\rfloor$, which corresponds to one of our blocks and the edges between it and the next block, the redomination number is at least one and we need at least three dominating elements. For convenience, we denote $G[V_i^4]\cup E^4_{i,i+1}$ simply by $G_i$. Similarly, we define $S_i := S \cap G_i$.

	Let $G_{i-1}$ and $G_i$ be two consecutive blocks. 
	Although $G_{i-1}$ and $G_i$ are disjoint, since they are connected, some elements of $G_i$ could be dominated by $S_{i-1}$.
	Depending on $S_{i-1}$, which we consider to be the emptyset when $i = 0$, one of the following six cases can happen. In each case, we show that every optimal mixed dominating set has a redomination number of at least one in $G_i$. This is illustrated in Figure~\ref{sixcases} below.
	
	\begin{center}
		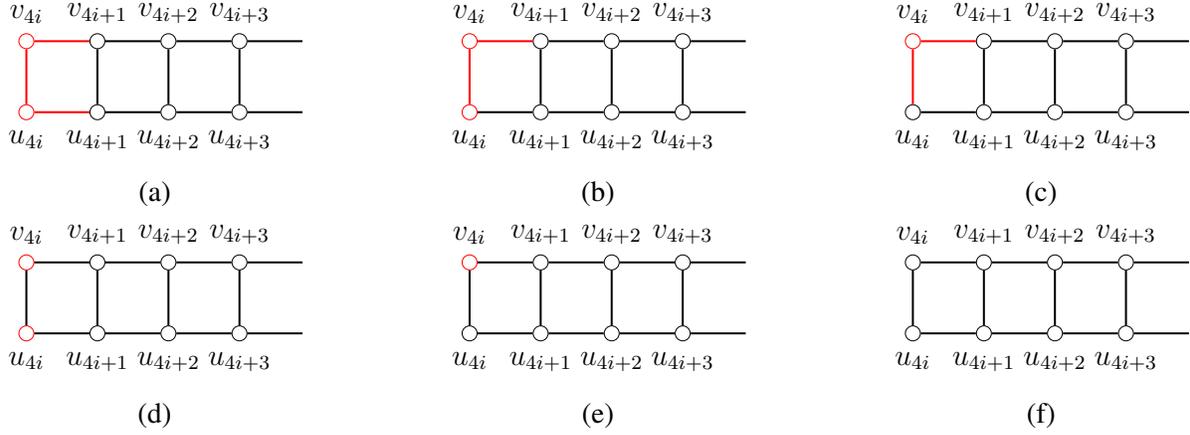
\begin{figure}[h]
			\begin{subfigure}[b]{.33\textwidth}
				\centering
				\resizebox*{.75\textwidth}{!}{
					\begin{tikzpicture}
					\vertex[label=$v_{4i}$,color=red](p1) at (0,0) {};
					\vertex[label=$v_{4i+1}$](p2) at (1,0) {};
					\vertex[label=$v_{4i+2}$](p3) at (2,0) {};
					\vertex[label=$v_{4i+3}$](p4) at (3,0) {};
					\vertex[color=white]  (p0) at (4,0) {};
					\vertex[label=below:$u_{4i}$,color=red](p5) at (0,-1) {};
					\vertex[label=below:$u_{4i+1}$](p6) at (1,-1) {};
					\vertex[label=below:$u_{4i+2}$](p7) at (2,-1) {};
					\vertex[label=below:$u_{4i+3}$](p8) at (3,-1) {}; 
					\vertex[color=white]  (p9) at (4,-1) {};
					\tikzset{EdgeStyle/.style={-}}
					\Edge[color=red](p1)(p2)
					\Edge(p2)(p3)
					\Edge(p3)(p4)
					\Edge(p4)(p0)
					\Edge[color=red](p1)(p5)
					\Edge(p2)(p6)
					\Edge(p3)(p7)
					\Edge(p4)(p8)
					\Edge[color=red](p5)(p6)
					\Edge(p6)(p7)
					\Edge(p7)(p8)
					\Edge(p8)(p9)
					\end{tikzpicture}}
				\caption{}    \label{a}            
			\end{subfigure}
			\begin{subfigure}[b]{0.33\textwidth}
				\centering
				\resizebox*{.75\textwidth}{!}{
					\begin{tikzpicture}
					\vertex[label=$v_{4i}$,color=red](p1) at (0,0) {};
					\vertex[label=$v_{4i+1}$](p2) at (1,0) {};
					\vertex[label=$v_{4i+2}$](p3) at (2,0) {};
					\vertex[label=$v_{4i+3}$](p4) at (3,0) {};
					\vertex[color=white]  (p0) at (4,0) {};
					\vertex[label=below:$u_{4i}$,color=red](p5) at (0,-1) {};
					\vertex[label=below:$u_{4i+1}$](p6) at (1,-1) {};
					\vertex[label=below:$u_{4i+2}$](p7) at (2,-1) {};
					\vertex[label=below:$u_{4i+3}$](p8) at (3,-1) {}; 
					\vertex[color=white]  (p9) at (4,-1) {};
					\tikzset{EdgeStyle/.style={-}}
					\Edge[color=red](p1)(p2)
					\Edge(p2)(p3)
					\Edge(p3)(p4)
					\Edge(p4)(p0)
					
					\Edge[color=red](p1)(p5)
					\Edge(p2)(p6)
					\Edge(p3)(p7)
					\Edge(p4)(p8)

					\Edge(p5)(p6)
					\Edge(p6)(p7)
					\Edge(p7)(p8)
					\Edge(p8)(p9)
					\end{tikzpicture}}
				\caption{}      \label{b}          
			\end{subfigure}
			\begin{subfigure}[b]{0.33\textwidth}
				\centering
				\resizebox*{.75\textwidth}{!}{
					\begin{tikzpicture}
					\vertex[label=$v_{4i}$,color=red](p1) at (0,0) {};
					\vertex[label=$v_{4i+1}$](p2) at (1,0) {};
					\vertex[label=$v_{4i+2}$](p3) at (2,0) {};
					\vertex[label=$v_{4i+3}$](p4) at (3,0) {};
					\vertex[color=white]  (p0) at (4,0) {};
					\vertex[label=below:$u_{4i}$](p5) at (0,-1) {};
					\vertex[label=below:$u_{4i+1}$](p6) at (1,-1) {};
					\vertex[label=below:$u_{4i+2}$](p7) at (2,-1) {};
					\vertex[label=below:$u_{4i+3}$](p8) at (3,-1) {}; 
					\vertex[color=white]  (p9) at (4,-1) {};
					\tikzset{EdgeStyle/.style={-}}
					\Edge[color=red](p1)(p2)
					\Edge(p2)(p3)
					\Edge(p3)(p4)
					\Edge(p4)(p0)
					
					\Edge[color=red](p1)(p5)
					\Edge(p2)(p6)
					\Edge(p3)(p7)
					\Edge(p4)(p8)

					\Edge(p5)(p6)
					\Edge(p6)(p7)
					\Edge(p7)(p8)
					\Edge(p8)(p9)
					\end{tikzpicture}}
				\caption{} \label{c}               
			\end{subfigure}
			\begin{subfigure}[b]{0.33\textwidth}
				\centering
				\resizebox*{.75\textwidth}{!}{
					\begin{tikzpicture}
					\vertex[label=$v_{4i}$,color=red](p1) at (0,0) {};                                               \vertex[label=$v_{4i+1}$](p2) at (1,0) {};
					\vertex[label=$v_{4i+2}$](p3) at (2,0) {};
					\vertex[label=$v_{4i+3}$](p4) at (3,0) {};
					\vertex[color=white]  (p0) at (4,0) {};
					\vertex[label=below:$u_{4i}$,color=red](p5) at (0,-1) {};
					\vertex[label=below:$u_{4i+1}$](p6) at (1,-1) {};
					\vertex[label=below:$u_{4i+2}$](p7) at (2,-1) {};
					\vertex[label=below:$u_{4i+3}$](p8) at (3,-1) {}; 
					\vertex[color=white]  (p9) at (4,-1) {};
					\tikzset{EdgeStyle/.style={-}}
					\Edge(p1)(p2)
					\Edge(p2)(p3)
					\Edge(p3)(p4)
					\Edge(p4)(p0)
					
					\Edge(p1)(p5)
					\Edge(p2)(p6)
					\Edge(p3)(p7)
					\Edge(p4)(p8)
					\Edge(p5)(p6)
					\Edge(p6)(p7)
					\Edge(p7)(p8)
					\Edge(p8)(p9)
					\end{tikzpicture}}
				\caption{} \label{d}             
			\end{subfigure}                           
			\begin{subfigure}[b]{0.33\textwidth}
				\centering
				\resizebox*{.75\textwidth}{!}{
					\begin{tikzpicture}
					\vertex[label=$v_{4i}$,color=red](p1) at (0,0) {};
					\vertex[label=$v_{4i+1}$](p2) at (1,0) {};
					\vertex[label=$v_{4i+2}$](p3) at (2,0) {};
					\vertex[label=$v_{4i+3}$](p4) at (3,0) {};
					\vertex[color=white]  (p0) at (4,0) {};
					\vertex[label=below:$u_{4i}$](p5) at (0,-1) {};
					\vertex[label=below:$u_{4i+1}$](p6) at (1,-1) {};
					\vertex[label=below:$u_{4i+2}$](p7) at (2,-1) {};
					\vertex[label=below:$u_{4i+3}$](p8) at (3,-1) {}; 
					\vertex[color=white]  (p9) at (4,-1) {};
					\tikzset{EdgeStyle/.style={-}}
					\Edge(p1)(p2)
					\Edge(p2)(p3)
					\Edge(p3)(p4)
					\Edge(p4)(p0)
					
					\Edge(p1)(p5)
					\Edge(p2)(p6)
					\Edge(p3)(p7)
					\Edge(p4)(p8)

					\Edge(p5)(p6)
					\Edge(p6)(p7)
					\Edge(p7)(p8)
					\Edge(p8)(p9)
					\end{tikzpicture}}
				\caption{}\label{e}
				\label{fig:LU-a2}
			\end{subfigure}           
			\begin{subfigure}[b]{0.33\textwidth}
				\centering
				\resizebox*{.75\textwidth}{!}{
					\begin{tikzpicture}
					\vertex[label=$v_{4i}$](p1) at (0,0) {};
					\vertex[label=$v_{4i+1}$](p2) at (1,0) {};
					\vertex[label=$v_{4i+2}$](p3) at (2,0) {};
					\vertex[label=$v_{4i+3}$](p4) at (3,0) {};
					\vertex[color=white]  (p0) at (4,0) {};
					\vertex[label=below:$u_{4i}$](p5) at (0,-1) {};
					\vertex[label=below:$u_{4i+1}$](p6) at (1,-1) {};
					\vertex[label=below:$u_{4i+2}$](p7) at (2,-1) {};
					\vertex[label=below:$u_{4i+3}$](p8) at (3,-1) {}; 
					\vertex[color=white]  (p9) at (4,-1) {};
					\tikzset{EdgeStyle/.style={-}}
					\Edge(p1)(p2)
					\Edge(p2)(p3)
					\Edge(p3)(p4)
					\Edge(p4)(p0)
					
					\Edge(p1)(p5)
					\Edge(p2)(p6)
					\Edge(p3)(p7)
					\Edge(p4)(p8)

					\Edge(p5)(p6)
					\Edge(p6)(p7)
					\Edge(p7)(p8)
					\Edge(p8)(p9)
					\end{tikzpicture}}
				\caption{}\label{f}
				\label{fig:LU-a2}
			\end{subfigure}            
			\caption{The six ways in which $S_{i-1}$ can dominate parts of $G_i$.}\label{sixcases}
		\end{figure}       
	\end{center}

	\begin{itemize}
		\item[(a)] The set $\{v_{4i},u_{4i},v_{4i}u_{4i},v_{4i}v_{4i+1},u_{4i}u_{4i+1}\}$ is dominated by some elements of $S_{i-1}$. To dominate $v_{4i+1}u_{4i+1}$, we must choose one element in $\{v_{4i+1},u_{4i+1},v_{4i+1}u_{4i+1},v_{4i+1}v_{4i+2},u_{4i+1}u_{4i+2}\}$. Any selection leads to at least one redomination, i.e.~ $rd_S(G_i)\geq 1$.
		
		\item[(b)]The set $\{v_{4i},u_{4i},v_{4i}u_{4i},v_{4i}v_{4i+1}\}$ is dominated by some elements of $S_{i-1}$. Each element that can dominate $v_{4i+1}u_{4i+1}$, increases the redomination number except $ u_{4i+2}u_{4i+3}$. By selecting this element, the vertex $v_{4i+1}$ remains undominated and then any choice to dominate it leads to some redomination.
		By symmetry, it is obvious that a similar result holds if $\{v_{4i},u_{4i},v_{4i}u_{4i},v_{4i}v_{4i+1}\} $ is dominated by $S_{i-1}$.
		
		\item[(c)] The set 
		$\{v_{4i},v_{4i}u_{4i},v_{4i}v_{4i+1}\}$ is dominated by $S_{i-1}$. By selecting  $u_{4i+1},v_{4i+2}v_{4i+3}$, every element of $G_i$  except $u_{4i+2}u_{4i+3},u_{4i+3},u_{4i+3}u_{4i+4}$ will be covered. Any selection for dominating $u_{4i+2}u_{4i+3}$, leads to redomination. It is clear that other choices also increase the redomination number.
		Again, by symmetry, one can show the same result for $\{u_{4i},v_{4i}u_{4i},u_{4i}u_{4i+1}\}$.
		
		\item[(d)]The set 
		$\{v_{4i},u_{4i}\}$ is dominated by $S_{i-1}$. To dominate $v_{4i}u_{4i}$, any selection  leads to redomination.
		
		\item[(e)]The set 
		$\{v_{4i}\}$ is dominated by $S_{i-1}$. To dominate $v_{4i}u_{4i}$, any selection except $u_{4i}u_{4i+1}$, increases the redomination number. However, $u_{4i}u_{4i+1}$ leads to a situation where the edge $v_{4i}v_{4i+1}$ is not dominated and then any selection to dominate this edge leads to redomination.
		
		\item[(f)]
		None of elements in $G_i$ is dominated by $S_{i-1}$. 
		$G_i$ contains $8$ vertices and $12$ edges and none of them are  dominated by previous block and every vertex in $P(n,1)$ has degree three.
		Every vertex in dominating set, dominates four vertices and three edges and every edge in dominating set, dominates two vertices and five edges.
		To dominate $G_i$ in this case, every edge dominates two vertices and five edges.
		We can easily investigate at least three dominating elements are needed and one redomination will be happen.
	\end{itemize}
	
	%The following results are concluded in the opposite inequality of proof.
	
	{\it Remark 1:}
	$\gmd(G_i)=3$ and $rd_S(G_i)\geq 1$.

	{\it Remark 2:}
	$\gmd(P(n,1))\geq 6m$ where $n=8m$ for some $m$.
	
	In the following discussion, we obtain the lower bound  for $\gmd(P(n,1))$ where $8\nmid n$.
	We consider a partitioning factor of $8$. In other words, we split $P(n,1)$ to consecutive blocks $V_0,V_1,\dots , V_{m-1}$ and the remained block $V_r$ (See Figure \ref{Blocks}).
	\begin{figure}[h!]
		\centering
		\includegraphics[scale=0.8]{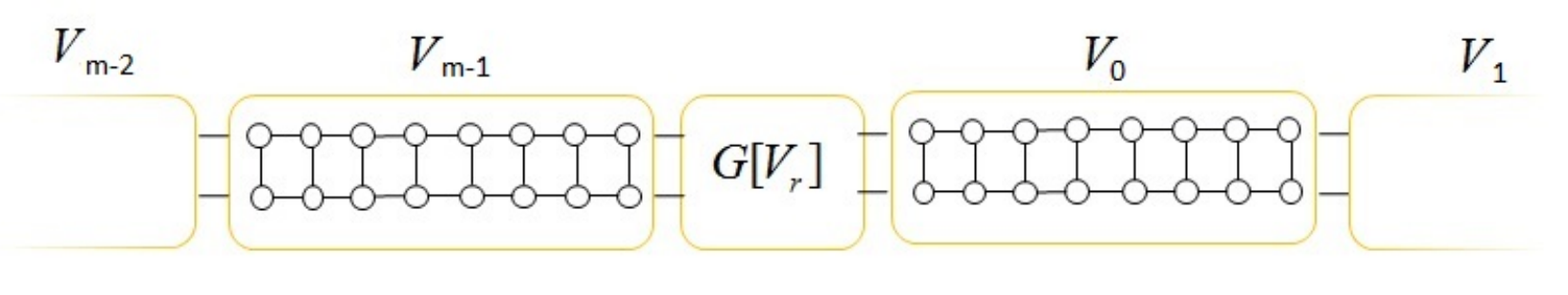}
		\caption{Blocks with  partitioning factor $8$ of $P(n,1)$ }\label{Blocks}
	\end{figure}

	We define the input and output elements of block $G[V_i^8]$, where $0\leq i\leq m-1$ (respectively, denoted by $I$ and $O$), as the set of elements that may be under the influence of the previous and next block. For instance, the input and output for block $V_0$ 
	are $I=\{v_0,u_0,v_0u_0,v_0v_1,u_0u_1\}$
	and 
	$O=\{v_7,u_7,v_7u_7,v_6v_7,u_6u_7\}$ that are specified by the red and blue colors, respectively (See Figure \ref{IO}).

	\begin{center}
		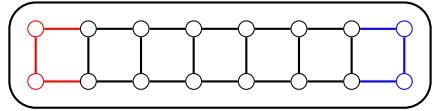
\begin{figure}[h!]
			\centering
			\begin{tikzpicture}[scale=0.7]
			
			\draw[thick,rounded corners=10pt] (-.5,.5) rectangle (7.5,-1.5);
			
			\vertex[color=red](p1) at (0,0) {};
			\vertex[](p2) at (1,0) {};
			\vertex[](p3) at (2,0) {};
			\vertex[](p4) at (3,0) {};
			\vertex[](p5) at (4,0) {};
			\vertex[](p6) at (5,0) {};
			\vertex[](p7) at (6,0) {};
			\vertex[color=blue](p8) at (7,0) {};
			
			\vertex[color=red](p1') at (0,-1) {};
			\vertex[](p2') at (1,-1) {};
			\vertex[](p3') at (2,-1) {};
			\vertex[](p4') at (3,-1) {};
			\vertex[](p5') at (4,-1) {};
			\vertex[](p6') at (5,-1) {};
			\vertex[](p7') at (6,-1) {};
			\vertex[color=blue](p8') at (7,-1) {};
			
			\tikzset{EdgeStyle/.style={-}}
			\Edge[color=red](p1)(p2)
			\Edge(p2)(p3)
			\Edge(p3)(p4)
			\Edge(p4)(p5)
			\Edge(p5)(p6)
			\Edge(p6)(p7)
			\Edge[color=blue](p7)(p8)
			
			\Edge[color=red](p1')(p2')
			\Edge(p2')(p3')
			\Edge(p3')(p4')
			\Edge(p4')(p5')
			\Edge(p5')(p6')
			\Edge(p6')(p7')
			\Edge[color=blue](p7')(p8')
			
			\Edge[color=red](p1)(p1')
			\Edge(p2)(p2')
			\Edge(p3)(p3')
			\Edge(p4)(p4')
			\Edge(p5)(p5')
			\Edge(p6)(p6')
			\Edge(p7)(p7')
			\Edge[color=blue](p8)(p8')            
			
			\end{tikzpicture}
			\caption{ 
				The input (\textcolor{red}{red}) and output (\textcolor{blue}{blue})} elements in a block\label{IO} 
		\end{figure}
	\end{center} 
	Our aim is to put a subset of $S_0\subseteq I$, as the initial mixed dominating set for block $V_i$ and find the effect of this input over output elements of the block when a mixed dominating set, like $S$, satisfies the following condition.
	
	\begin{itemize}
		\item[1)]$S_0\subseteq S$.
		
		\item[2)]$X=\{\upxi \in G[V_i^8] \mid |N_{V\cup E}[\upxi]|=7\}$ must be dominated.
		
		\item[3)]At most one re-dominating occurs by dominating set $S$.
	\end{itemize}
	
	For convenience without loss of generality, we discuss on block $V_0$.  The condition $3$  leads to $S_0$  be one of the following sets $$\{v_0\},\{u_0\},\{v_0u_0\},\{v_0v_1\},\{u_0u_1\}.$$

	By  brute-force search, easily investigate that the output for every input $S_0$ of the above sets is as follow:
	
	\begin{itemize}
		\item[$\bullet$]
		$S_0=\{v_0\},$ the output is $\{u_6u_7\}$,
		\item[$\bullet$]
		$S_0=\{u_0\},$ the output is $\{v_6v_7\}$,
		\item[$\bullet$]
		$S_0=\{v_0u_0\},$ according to condition $2$, it is impossible,
		\item[$\bullet$]
		$S_0=\{v_0v_1\},$ the output is $\{u_7\}$ or $\emptyset$,
		\item[$\bullet$]
		$S_0=\{u_0u_1\},$ the output is $\{v_7\}$ or $\emptyset$.
	\end{itemize}
	
	In the following, our aim is to glue $G[V_0^8]$ and $G[V_1^8]$  and prove that corresponding to every input, the output of $G[V_{1}^8]$ is similar to the output of $G[V_0^8]$. 
%	\todo{$G[V_{i+1}^8]$ is similar to the output of $G[V_i^8]$. nistesh?}
	% On the other hand, output of $G[V_i]$ should be compatible with the input $G[V_{i+1}]$. 
	%The proof of above claim as follows

	Let $S_0=\{u_0\}(or \; \{v_0\})$, then the resulted output is $\{u_6u_7\} (\{v_6v_7\})$. Therefor none of  the elements in $G[V_1^8] \cup \{u_7,u_7u_8\}(G[V_1^8] \cup \{v_7,v_7v_8\})$ are dominated. So the best choice for mixed dominating set such that satisfy condition $2$, is $\{u_8\}(\{v_8\})$. This choice leads to  $G[V_1^8]$ has the same input as $G[V_0^8]$. Therefore the output of $G[V_0^8]$ and $G[V_1^8]$ are similar.
	
	Let $S_0=\{v_0v_1\}$ and the output be $\{u_7\}$. Therefore none of  the elements in $G[V_1^8]\setminus\{v_8\} \cup \{v_7v_8\}$ are dominated. So the best choice for every optimal mixed dominating set such that satisfy condition $2$, is $\{v_8v_{9}\}$. This choice leads to that $G[V_1^8]$ has the same input as $G[V_0^8]$.
	
	In the case $S_0=\{v_0v_1\},$ and the output is $\emptyset$ none of  the elements in $G[V_1^8] \cup \{v_7v_8,u_6u_7,u_7u_8,u_7\}$ are dominated. So the best choice for every  optimal mixed dominating set such that satisfy condition $2$, is $\{u_7u_8,v_8v_9\}$. This choice leads to that $G[V_1^8]$ has the same input as $G[V_0^8]$ and therefor the output is same.
	
	For $S_0=\{u_0u_1\},$ the proof is similar the the case $S_0=\{v_0v_1\}$.

	The above discussion is for two consecutive block. By inductively, we can extend the results to arbitrary number of consecutive blocks. 
	
	Now, we summarize the elements of $P(n,1)$ which are not dominated in the above process, in the following table:
	\begin{center}
		\begin{table}[h]
			\caption{Not Dominated Elements}\label{Not dominated Elements}
			\centering
			\begin{tabular}{|c|c|c|}
				\hline
				Input & Output & Not Dominated elements \\ 
				\hline
				$\{v_0\}$ & $\{u_{8m-2}u_{8m-1}\}$ & $ \{v_{8m-1}\} \cup G[V_r]\setminus \{v_{n-1}\}   $    \\
				\hline
				$\{u_0\}$ & $\{v_{8m-2}v_{8m-1}\}$ & $ \{u_{8m-1}\} \cup G[V_r]\setminus \{u_{n-1}\}   $    \\
				\hline
				$\{v_0v_1\}$ & $\{u_{8m-1}\}$ & $\{v_{8m-1}v_{8m},u_{n-1}u_0,u_0\}  \cup G[V_r]\setminus \{u_{8m}\}   $    \\
				\hline
				$\{v_0v_1\}$ & $\{\}$ &$\{u_{8m-2}u_{8m-1},v_{8m-1}u_{8m-1},u_{8m-1}u_{8m},u_{8m-1},v_{8m-1}v_{8m},u_0u_{n-1},u_0\}  \cup G[V_r]   $   \\
				\hline
				$\{u_0u_1\}$ & $\{v_{8m-1}\}$ & $\{u_{8m-1}u_{8m},v_0v_{n-1},v_0\}  \cup G[V_r]\setminus \{v_{8m}\}   $    \\
				\hline
				$\{u_0u_1\}$ & $\{\}$ &$\{v_{8m-2}v_{8m-1},u_{8m-1}v_{8m-1},v_{8m-1}v_{8m},v_{8m-1},u_{8m-1}u_{8m},v_0v_{n-1},v_0\}  \cup G[V_r]   $   \\	
				\hline 
			\end{tabular}  	 
		\end{table}
	\end{center}
	To dominate elements that are appeared in the third column of the above table, it is sufficient to know how many elements are needed to dominated remaining elements, is denoted by $\# rem(G)$ as shown Table \ref{number need}.

	\begin{center}
		\begin{table}[h]
			\caption{number of needed elements  to dominate remaining elements}\label{number need}
			\centering
			\begin{tabular}{|c|c|c|c|c|c|c|c|}
				\hline
				$n \mod 8$ & 1 & 2 & 3 & 4 & 5 & 6& 7\\
				\hline
				$\#rem(G)$ & 2 & 2 & 3 & 4 & 4 & 5 & 6 \\
				\hline
			\end{tabular}  	 
		\end{table}
	\end{center}
	Therefore  $\gmd(P(n,1))$ is at least equal to the given values in Equation  (\ref{Eq-1}).
\end{proof}

\subsection{The Case k = 2}
Let $t=4$, So $m= \lfloor \frac{n}{4} \rfloor $ and $\bar{r}=n \mod 4$.
\begin{theorem}
	Mixed domination number of $P(n,2)$ where $n\geq 4$ is
	\[
	\gmd(P(n,2))=\left\{\begin{array}{lll}
	3m & \text{if } & \bar{r}=0,  \\ \noalign{\medskip}
	3m+1 & \text{if }& \bar{r}=1,  \\ \noalign{\medskip}
	3m+2 & \text{if }& \bar{r}=2,  \\ \noalign{\medskip}
	3m+3 & \text{if }& \bar{r}=3,  \\ \noalign{\medskip}
	\end{array}\right.
	\]

\end{theorem}
\begin{proof}
	We consider two cases:
	
	{\it Case 1:} $\gmd(P(n,2))$ is at most equal to the given values and it is clear because we can construct a mixed dominating set $S$ as the following:
	
	\[
	S=\left\{\begin{array}{lll}
	M & \text{if } & \bar{r}=0,  \\ \noalign{\medskip}
	M\cup \{v_{4m+1}u_{4m+1}\} & \text{if }& \bar{r}=1,  \\ \noalign{\medskip}
	M\cup \{v_{4m+1}u_{4m+1},v_{4m+2}u_{4m+2}\} & \text{if }& \bar{r}=2,  \\ \noalign{\medskip}
	M\cup \{v_{4m+1}u_{4m+1},v_{4m+2}u_{4m+2},v_{4m+3}u_{4m+3}\} & \text{if }& \bar{r}=3,  \\ \noalign{\medskip}
	\end{array}\right.
	\]
	where \[M=\{v_{4i}u_{4i},u_{4i+1}u_{4i+3},v_{4i+2} :  0\leq i\leq m-1\}.\]

	In Figure \ref{fig-P(n,2)}, we show the dominating sets of $P(n, 2)$ for $8 \leq n \leq 11$, where the domination elements of $S$ are in dark.
	\begin{center}
		\begin{figure}[h]
			\begin{subfigure}[b]{0.23\textwidth}
				\centering
				\resizebox*{.8\textwidth}{!}{
					\begin{tikzpicture}[rotate=90]
					\tikzstyle{VertexStyle} = [shape = classic]
					\tikzstyle{EdgeStyle}= [thin,double= black,double distance= 2pt]
					\SetVertexNoLabel
					\SetVertexNormal[LineWidth=1pt]
					\grGeneralizedPetersen[Math,RA=5]{8}{2}
					\AddVertexColor{black}{a2,a6}
					\draw[draw=black,line width=4pt]
					(a0) to (b0);
					\draw[draw=black,line width=4pt]
					(a4) to (b4);
					\draw[draw=black,line width=4pt]
					(b1) to (b3);
					\draw[draw=black,line width=4pt]
					(b5) to (b7); 
					\end{tikzpicture}}
				\caption{$P(8,2)$}             
			\end{subfigure}
			\begin{subfigure}[b]{0.23\textwidth}
				\centering
				\resizebox*{.8\textwidth}{!}{
					\begin{tikzpicture}[rotate=90]
					\tikzstyle{VertexStyle} = [shape = classic]
					\tikzstyle{EdgeStyle}= [thin,double= black,double distance= 2pt]
					\SetVertexNoLabel
					\SetVertexNormal[LineWidth=1pt]
					\grGeneralizedPetersen[Math,RA=5]{9}{2}
					\AddVertexColor{black}{a7,a3}
					\draw[draw=black,line width=4pt]
					(a0) to (b0);
					\draw[draw=black,line width=4pt]
					(a5) to (b5);
					\draw[draw=black,line width=4pt]
					(b8) to (b6);
					\draw[draw=black,line width=4pt]
					(b4) to (b2);
					\draw[draw=black,line width=4pt]
					(a1) to (b1);
					\end{tikzpicture}}
				\caption{$P(9,2)$}          
			\end{subfigure}
			\begin{subfigure}[b]{0.23\textwidth}
				\centering
				\resizebox*{.8\textwidth}{!}{
					\begin{tikzpicture}[rotate=90]
					\tikzstyle{VertexStyle} = [shape = classic]
					\tikzstyle{EdgeStyle}= [thin,double= black,double distance= 2pt]
					\SetVertexNoLabel
					\SetVertexNormal[LineWidth=1pt]
					\grGeneralizedPetersen[Math,RA=5]{10}{2}
					\AddVertexColor{black}{a8,a4}
					\draw[draw=black,line width=4pt]
					(a0) to (b0);
					\draw[draw=black,line width=4pt]
					(a6) to (b6);
					\draw[draw=black,line width=4pt]
					(b9) to (b7);
					\draw[draw=black,line width=4pt]
					(b5) to (b3);
					\draw[draw=black,line width=4pt]
					(a2) to (b2);
					\draw[draw=black,line width=4pt]
					(a1) to (b1);
					\end{tikzpicture}}
				\caption{$P(10,2)$} \label{c}               
			\end{subfigure}
			\begin{subfigure}[b]{0.23\textwidth}
				\centering
				\resizebox*{.8\textwidth}{!}{
					\begin{tikzpicture}[rotate=90]
					\tikzstyle{VertexStyle} = [shape = classic]
					\tikzstyle{EdgeStyle}= [thin,double= black,double distance= 2pt]
					\SetVertexNoLabel
					\SetVertexNormal[LineWidth=1pt]
					\grGeneralizedPetersen[Math,RA=5]{11}{2}
					\AddVertexColor{black}{a9,a5}
					\draw[draw=black,line width=4pt]
					(a0) to (b0);
					\draw[draw=black,line width=4pt]
					(a7) to (b7);
					\draw[draw=black,line width=4pt]
					(b10) to (b8);
					\draw[draw=black,line width=4pt]
					(b6) to (b4);
					\draw[draw=black,line width=4pt]
					(a3) to (b3);
					\draw[draw=black,line width=4pt]
					(a2) to (b2);
					\draw[draw=black,line width=4pt]
					(a1) to (b1);
					\end{tikzpicture}}                                       
				\caption{$P(11,2$}                
			\end{subfigure}
			\caption{Mixed dominating set of $P(n,2)$, where $8\leq n \leq 11$}\label{fig-P(n,2)}
		\end{figure}
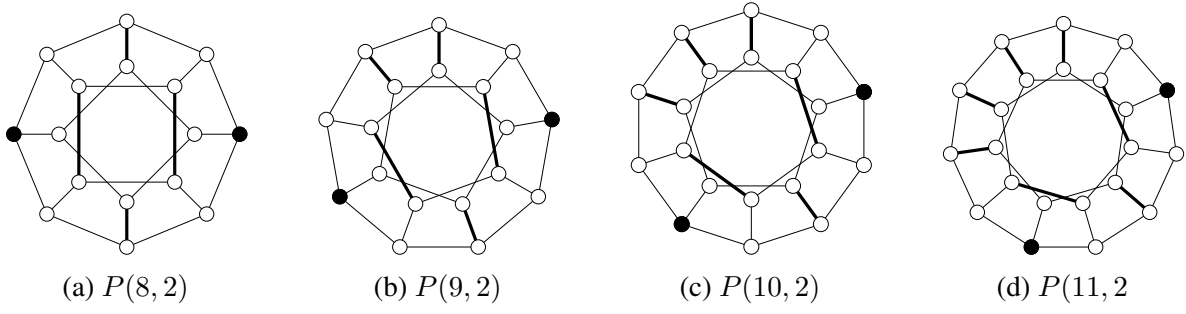       
	\end{center}

	{\it Case 2:} $\gmd(P(n,2))$ is at least equal to the given values.
	To prove this claim we need the following lemma.
	
	\begin{lemma}\label{Lemma-p(n,2)}
		Let $S$ be an arbitrary mixed dominating set of $P(n, 2)$, then for every 4-block $G_i$ ,where $0 \leq i \leq m-1 $, we have $rd_S(G_i)\geq 1.$ 
	\end{lemma}
	\begin{proof}
		Let $G_{i-1},G_i$ be two consecutive blocks. For every mixed dominating set like $S$, we define $S_i=S \cap G_i.$ 
		
		As the blocks $G_{i-1}$ and $G_i$ are connected, then some elements in $S_{i-1}$ may be dominated by some elements in $G_i$. Independently of which elements of $V_{i-1}$ included in $S_{i-1}$, one of the 27 cases which are appeared in Table \ref{Not dominated Elements for 2} will be occurred. As Theorem \ref{Domination(n,1)},  we can easily investigate that every minimum mixed dominating set have at least one re-dominating. For example, let input be  $$I_1=\{v_{4i},u_{4i},v_{4i}u_{4i},v_{4i}v_{4i+1},u_{4i}u_{4i+2},u_{4i+1},u_{4i+1}u_{4i+3},v_{4i+1}u_{4i+1}\},$$ 
		
		 then every elements except $v_{4i+2}$ leads increasing re-domination number and if $v_{4i+2}$ is selected, then increasing in the re-domination number occurs in the remaining edges.
		
		When the input be $I_2, \cdots, I_{27}$, the discussion is similar as $I_1$.
	\end{proof}

	\begin{center}
		\begin{table}[!h]
			\caption{Not Dominated Elements}\label{Not dominated Elements for 2}
			\centering
			\resizebox{\textwidth}{!}{
				\begin{tabular}{ll}
					
					$I_{1}=\{v_{4i}, u_{4i}, u_{4i+1}, v_{4i}v_{4i+1}, v_{4i}u_{4i}, u_{4i}u_{4i+2}, u_{4i+1}u_{4i+3}, v_{4i+1}u_{4i+1}\}$ & $I_{15}=\{v_{4i}, u_{4i}\}$
					\\
					$I_{2}=\{v_{4i}, u_{4i}, u_{4i+1}, v_{4i}v_{4i+1}, v_{4i}u_{4i}, u_{4i}u_{4i+2}\}$ & $I_{16}=\{v_{4i},  u_{4i+1}, u_{4i+1}u_{4i+3}, v_{4i+1}u_{4i+1}\}$
					\\
					$I_{3}=\{v_{4i}, u_{4i}, v_{4i}v_{4i+1}, v_{4i}u_{4i}, u_{4i}u_{4i+2}\}$ &  $I_{17}=\{v_{4i}, u_{4i}\}$
					\\
					$I_{4}=\{v_{4i}, u_{4i}, u_{4i+1}, v_{4i}v_{4i+1}, v_{4i}u_{4i},  u_{4i+1}u_{4i+3}, v_{4i+1}u_{4i+1}\}$ & $I_{18}=\{v_{4i}\}$
					\\
					$I_{5}=\{v_{4i}, u_{4i}, u_{4i+1}, v_{4i}v_{4i+1}, v_{4i}u_{4i}\}$ & $I_{19}=\{ u_{4i}, u_{4i+1},  v_{4i}u_{4i}, u_{4i}u_{4i+2}, u_{4i+1}u_{4i+3}, v_{4i+1}u_{4i+1}\}$
					\\
					$I_{6}=\{v_{4i}, u_{4i},  v_{4i}v_{4i+1}, v_{4i}u_{4i}, u_{4i}u_{4i+2}\}$ &   $I_{20}=\{ u_{4i}, u_{4i+1},  v_{4i}u_{4i}, u_{4i}u_{4i+2}\}$
					\\
					$I_{7}=\{v_{4i},  u_{4i+1}, v_{4i}v_{4i+1}, v_{4i}u_{4i}, u_{4i}u_{4i+2},  v_{4i+1}u_{4i+1}\}$ & $I_{21}=\{ u_{4i},   v_{4i}u_{4i}, u_{4i}u_{4i+2}\}$
					\\
					$I_{8}=\{v_{4i},  u_{4i+1}, v_{4i}v_{4i+1}, v_{4i}u_{4i}, \}$ & $I_{22}=\{ u_{4i}, u_{4i+1},  u_{4i+1}u_{4i+3}, v_{4i+1}u_{4i+1}\}$
					\\
					$I_{9}=\{v_{4i},  v_{4i}v_{4i+1}, v_{4i}u_{4i}, \}$ & $I_{23}=\{ u_{4i}, u_{4i+1}\}$
					\\
					$I_{10}=\{v_{4i}, u_{4i}, u_{4i+1},  v_{4i}u_{4i}, u_{4i}u_{4i+2}, u_{4i+1}u_{4i+3}, v_{4i+1}u_{4i+1}\}$&
					$I_{24}=\{ u_{4i}\}$
					\\
					$I_{11}=\{v_{4i}, u_{4i}, u_{4i+1},  v_{4i}u_{4i}, u_{4i}u_{4i+2}\}$  & 
					$I_{25}=\{  u_{4i+1}, u_{4i+1}u_{4i+3}, v_{4i+1}u_{4i+1}\}$
					\\
					$I_{12}=\{v_{4i}, u_{4i},   v_{4i}u_{4i}, u_{4i}u_{4i+2}\}$ &
					$I_{26}=\{ u_{4i+1}\}$ \\
					$I_{13}=\{v_{4i}, u_{4i}, u_{4i+1},  u_{4i+1}u_{4i+3}, v_{4i+1}u_{4i+1}\}$& 
					$I_{27}=\{ \}$
					\\
					$I_{14}=\{v_{4i}, u_{4i}, u_{4i+1}\}$&     
					\\
					
				\end{tabular}  }	 
			\end{table}
		\end{center}
		
		\begin{corollary}
			Let $n$ be the multiple of four. Then $\gmd(P(n,2))=\frac{3n}{4}$.
		\end{corollary}
		
		In the following, we split $P(n,2)$ to consecutive 4-blocks $V_0,V_1,\dots , V_{m-1}$ 
	and $V_r$ where  $V_r$ is the remaining block (see Figure \ref{IO-fig}).
		\begin{figure}[h!]
			\centering
			\includegraphics[scale=0.8]{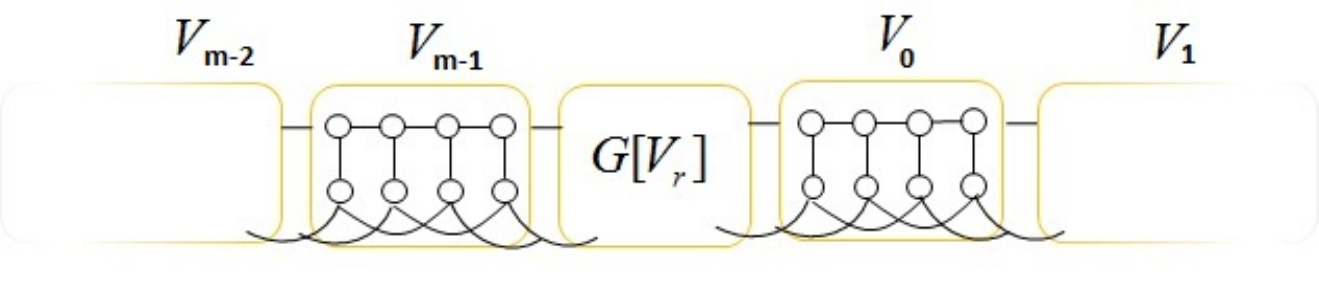}
			\caption{Blocks with  partitioning factor $4$ of $P(n,2)$ }\label{IO-fig}
		\end{figure}
		Each blocks of $V_0, \cdots, V_{m-1}$ has one re-dominating number that can dominate an extra edge.
		According to this  fact, it is possible at most two of edges in remaining block can be covered by the blocks $V_0$ and $V_{m-1}$.
		We check three possible case that $G[V_r]$ can be get.
		
		{\it a) n=4k+1:}
		The block $V_r$ contains two vertices and three edges. At most two of its edges can be covered by blocks $V_0$ and $V_{m-1}$, so we need one more vertex or edge to cover remaining elements of $G[V_r]$.
		
		{\it b) n=4k+2:}
		The block $V_r$ contains four vertices and six edges, so we need two more elements  to cover remaining elements.
		
		{\it c) n=4k+3:}
		The block $V_r$ contains six vertices and nine edges, so we need three more elements  to cover remaining elements.
		
		Therefore, the proposed  lower bound in case 2 can be concluded.
	\end{proof}
	
	\begin{remark}
		It is obvious that there exist different patterns which by using them, we can find minimum mixed dominating set for some $n$ in $P(n,2)$. For example, let $m= \lfloor \frac{n}{8} \rfloor $ and $\bar{r}=n \mod 8$, so we can  split graph into 8-blocks and select the below set in each block.
		\[M=\{u_{8i},v_{8i+1}v_{8i+2},u_{8i+3},v_{8i+4}u_{8i+4},v_{8i+5}v_{8i+6},v_{8i+7}u_{8i+7} :  0\leq i\leq m-1\}.\]
		So, the mixed domination number of $(P(n,2))$ is
		\[
		\gmd(P(n,2))=\left\{\begin{array}{lll}
		6m & \text{if } & \bar{r}=0,  \\ \noalign{\medskip}
		6m+2 & \text{if }& \bar{r}=1,2,  \\ \noalign{\medskip}
		6m+3 & \text{if }& \bar{r}=3,  \\ \noalign{\medskip}
		6m+4 & \text{if }& \bar{r}=4,5,  \\ \noalign{\medskip}
		6m+5 & \text{if }& \bar{r}=6,  \\ \noalign{\medskip}
		6m+6 & \text{if }& \bar{r}=7,  \\ \noalign{\medskip}
		\end{array}\right.
		\]
		It is clear that this pattern is correct for all $n$ except when $\bar{r}={1,4}$ and for these values the previous pattern is suitable.  
	\end{remark}

	\subsection{The Case $k\geq 3:$ }
	In this section, we generally present  an upper bound for $\gmd(P(n,k))$ and also we will propose a conjecture that this bound is the exact value of  mixed domination number of generalized Petersen graph. 
	
	For convenience, we let $T=4k'+1$ where $k=2k'$ or $k=2k'+1$,
	also $m=\lfloor \frac{n}{T} \rfloor$ and $\bar{r}=n \mod T$. We split graph to blocks $V_0,V_1,\cdots, V_{m-1} $ with partitioning factor $T$ and the remaining of graph by $V_r$.
	
	\begin{theorem}
		Let $k\geq 3$. The mixed domination number of $P(n,k)$ is
		\[
		\gmd(P(n,k))\leq \left\{\begin{array}{lll}
		(3k'+1).m + \bar{r} & \text{if } & \bar{r} \text{  is evev},  \\ \noalign{\medskip}
		
		(3k'+1).m + k' +  \frac{\bar{r}+1}{2} & \text{if }&  \bar{r} \text{  is odd}.  \\ \noalign{\medskip}
		\end{array}\right.
		\]
	\end{theorem}
	\begin{proof}
		It is sufficient to construct a mixed domination with this bound.
		To show these inequalities, we consider two cases.

%		{\it Case 1:} $\bar{r}=0$
%		
%		In this case, for each block we define the set $M$ as
%		\todo{avali ezafast  pakesh kardam}
%		\[M=\{u_{2i},v_{2i+1}u_{2i+1},v_{2k'+2i}v_{2k'+2i+1} :  0\leq i < k'\}\cup \{v_{4k'}u_{4k'}\},\]
%		
%		where $k$ is even and for odd $k$,
%		\todo{avali ezafast  pakesh kardam}
%		\[M=\{u_{2i+1},v_{2i+2}u_{2i+2},v_{2k'+2i+1}v_{2k'+2i+2} :  0\leq i < k'\}\cup \{v_{0}u_{0}\}.\]
%		\todo{behtar nist benevisim? age movafegi bayad kol M haye magale dorost beshe!!!!}
		 {\it Case 1:} $\bar{r}=0$
				
		In this case, for block $j$ where $0\leq j\leq m-1$ we define the set $M_j$ as
		\[M_j=\{u_{T\times j + 2i},v_{T\times j + 2i+1}u_{T\times j + 2i+1},v_{T\times j + 2k'+2i}v_{T\times j + 2k'+2i+1} :  0\leq i < k'\}\cup \{v_{T\times j + 4k'}u_{T\times j + 4k'}\},\]
		where $k$ is even and for odd $k$,
		\[M_j=\{u_{T\times j + 2i+1},v_{T\times j + 2i+2}u_{T\times j + 2i+2},v_{T\times j + 2k'+2i+1}v_{T\times j + 2k'+2i+2} :  0\leq i < k'\}\cup \{v_{T\times j }u_{T\times j  }\}.\]
		
		It is easy to see that $S$ is the mixed dominating set and  each block has  $3k'+1$ elements in mixed dominating set. In the following figures we show a mixed dominating sets for $P(27,4)$ and $P(27,5)$.

		\begin{center}\label{P(27,4-5)}
			\resizebox{.60\textwidth}{!}{
				$
				\begin{array}{ccc}
				\begin{tikzpicture}[ every node/.style={draw,circle,inner sep=0pt}] %,minimum size=5mm
				\graph[clockwise, radius=8cm] {subgraph C_n [n=27,name=A] };
				\graph[clockwise, radius=6cm] {subgraph I_n [n=27,name=B] };
				\foreach \i 
				in {1,2,3,4,5,6,7,8,9,10,11,12,13,14,15,16,17,18,19,20,21,22,23,24,25,26,27}{
					\AddVertexColor{white}{A \i}
					\AddVertexColor{white}{B \i}
				}
				
				\foreach \i [evaluate={\j=int(mod(\i+3,27)+1)}]
				in {1,2,3,4,5,6,7,8,9,10,11,12,13,14,15,16,17,18,19,20,21,22,23,24,25,26,27}{
					\draw (A \i) -- (B \i);
					\draw (B \j) -- (B \i);
				}
				\foreach \i 
				in {1,3,10,12,19,21}{
					\AddVertexColor{black}{B \i}}
				
				\foreach \i 
				in {2,4,9,11,13,18,20,22,27}{
					\draw[draw=black,line width=3pt] (A \i) -- (B \i);	}
				\foreach \i [evaluate={\j=int(mod(\i+1,27))}]
				in {5,7,14,16,23,25}{
					\draw[draw=black,line width=3pt] (A \i) -- (A \j);	}
				\end{tikzpicture}
				& \,\,
				&
				\begin{tikzpicture}[rotate=1,every node/.style={draw,circle,inner sep=0pt}] %,minimum size=5mm
				\graph[clockwise, radius=8cm] {subgraph C_n [n=27,name=A] };
				\graph[clockwise, radius=6cm] {subgraph I_n [n=27,name=B] };
				\foreach \i 
				in {1,2,3,4,5,6,7,8,9,10,11,12,13,14,15,16,17,18,19,20,21,22,23,24,25,26,27}{
					\AddVertexColor{white}{A \i}
					\AddVertexColor{white}{B \i}
				}
				\foreach \i 
				in {2,4,11,13,20,22}{
					\AddVertexColor{black}{B \i}}
				\foreach \i [evaluate={\j=int(mod(\i+4,27)+1)}]
				in {1,2,3,4,5,6,7,8,9,10,11,12,13,14,15,16,17,18,19,20,21,22,23,24,25,26,27}{
					\draw (A \i) -- (B \i);
					\draw (B \j) -- (B \i);
				}	
				
				\foreach \i 
				in {1,3,5,10,12,14,19,21,23}{
					\draw[draw=black,line width=3pt] (A \i) -- (B \i);	}
				\foreach \i [evaluate={\j=int(mod(\i+1,28))}]
				in {6,8,15,17,24,26}{
					\draw[draw=black,line width=3pt] (A \i) -- (A \j);	}
				\end{tikzpicture}
				\\
				\\
				
				${\Huge P(27,4)}$ &\,\, &  ${\Huge P(27,5)}$  \\
				
				\end{array} 
				$
			}
		\end{center}
		
		{\it Case 2:} $\bar{r}\neq 0$
		
		For each block $V_0,V_1,\dots ,V_{m-1}$, we set a mixed dominating set similar to case 1 and for remaining block $V_r$ according to the values of $k$  and $\bar{r}$, we add appropriate elements to mixed dominating set as in the tables \ref{$V^r1$} and \ref{$V^r2$}:
%		\todo{2 ta jadvalo kolan avaz kardam anishasho va bar hasb kol graph neveshtam. chekesh kon. dar zemn onvanesham dorost nist chon anasor 2 ta block akhar bayad poshesh dade beshan}
\begin{center}
	\begin{table}[h!]
		\caption{Elements needed for dominating undominated elements in $V_r$ and $V_{m-1}$ when $\bar{r}\leq 2k'$}\label{$V^r1$}
		\centering
		\resizebox{\textwidth}{!}{
		\begin{tabular}{|l|c|c|} 		
			
			\cline{1-3} $k$ is even & $\bar{r}$ is even &
			$\{u_{T\times m + 2i},v_{T\times m + 2i+1}u_{T\times m + 2i+1} :  0\leq i< \frac{\bar{r}}{2}\} $ \\ 
			\cline{2-3}  &  $\bar{r}$ is odd  & $\{u_{T\times m + 2i},v_{T\times m + 2i+1}u_{T\times m + 2i+1} :  0\leq i < \frac{\bar{r}-1}{2}\}\cup \{v_{T\times m + \bar{r}-1}u_{T\times m + \bar{r}-1}\}\cup \{u_{T\times m-2i-2}:0\leq i < \frac{2k'-\bar{r}-1}{2}\}$\\ 
			\cline{1-3} $k$ is odd & $\bar{r}$ is even  &  $\{u_{T\times m + 2i+1},v_{T\times m + 2i}u_{T\times m + 2i} :  0\leq i < \frac{\bar{r}}{2}\}$\\ 
			\cline{2-3}  & $\bar{r}$ is odd &  $\{u_{T\times m + 2i+1},v_{T\times m + 2i}u_{T\times m + 2i} :  0\leq i < \frac{\bar{r}-1}{2}\}\cup \{v_{T\times m + \bar{r}-1}u_{T\times m + \bar{r}-1}\}\cup \{u_{T\times m-2i-2}: 0\leq i < \frac{2k'-\bar{r}-1}{2}\}$  \\
			\cline{1-3} 
		\end{tabular} 
	}
	\end{table}
\end{center}

\begin{center}
	\begin{table}[h!]  
		\caption{Elements needed for dominating undominated elements in $V_r$ and $V_{m-1}$ when $\bar{r}> 2k'$}\label{$V^r2$}
		\centering
		\resizebox{\textwidth}{!}{
		\begin{tabular}{|l|c|c|} 		
			\cline{1-3} $k$ is even & $\bar{r}$ is even &  $\{u_{T\times m + 2i},v_{T\times m + 2i+1}u_{T\times m + 2i+1} :  0\leq i < \frac{\bar{r}}{2}\} $\\ 
			\cline{2-3}  &  $\bar{r}$ is odd  &
			$\{u_{T\times m + 2i},v_{T\times m + 2i+1}u_{T\times m + 2i+1} :  0\leq i < k'\}\cup \{v_{T\times m + 2k'+2i}v_{T\times m + 2k'+2i+1}: 0\leq i < \frac{\bar{r}-2k'+1}{2}\}$	\\ 
			\cline{1-3} $k$ is odd & $\bar{r}$ is even  &  $\{u_{T\times m + 2i+1},v_{T\times m + 2i}u_{T\times m + 2i} :  0\leq i < \frac{\bar{r}}{2}\}$ \\ 
			\cline{2-3}  & $\bar{r}$ is odd &  $\{u_{T\times m + 2i+2},v_{T\times m + 2i+1}u_{T\times m + 2i+1} :  0\leq i < k'\}\cup \{v_{T\times m }u_{T\times m }\}\cup \{v_{T\times m + 2k'+2i+1}v_{T\times m + 2k'+2i+2}: 0\leq i < \frac{\bar{r}-2k'-1}{2}\}$ \\
			\cline{1-3} 
		\end{tabular} 
	}
	\end{table}
\end{center}
		
	\end{proof}
\section{Conclusion and open problems} \label{sec:conclusion}
In this paper, we investigated the exact value of $\gmd(P(n,1))$ and $\gmd(P(n,2))$ and also we  proved their correctness. The  upper bounds for $\gmd(P(n, k))$ where $k\geq 3$ is proposed.  According to the proposed upper bound, we gave a conjecture that this upper bound  is exact value for $\gmd(P(n, k))$.

As a further study, it is interesting to answer the following questions:
\begin{enumerate}
	\item[-]Exact value of $\gmd(P(n, k))$ where $k\geq 3$.
	\item[-]Enumerate the number of different mixed dominating set in Petersen graphs.
	\item[-]As the generalized Petersen graphs are a particular cases of $\mathcal{I}-$graphs. It is also interesting to find out the exact value of the mixed domination number of $\mathcal{I}-$graphs.
	\item[-]Design constructive or algorithmic method to find $\gmd$ for other classes of graph like intervals, permutation graphs and etc.
\end{enumerate}

\section*{Acknowledgment}
The authors are grateful to A. K. Goharshady for their constructive comments and suggestions on improving our paper.

%\bibliographystyle{acm}
%\bibliography{mybib}

\end{document}